\documentclass[10pt]{article}

\usepackage{booktabs} 
\usepackage[utf8]{inputenc}
\usepackage[T1]{fontenc}  
\usepackage[sc]{mathpazo}
\usepackage{subcaption}
\usepackage{amsmath,amssymb}
\usepackage{amsfonts}
\usepackage{bm}
\usepackage{epic}
\usepackage{rotating}
\usepackage{csquotes}
\usepackage{graphicx}
\usepackage{tikz}
\usepackage[mathscr]{eucal}
\usepackage{amsthm}
\usepackage{xspace}
\usepackage{psfrag}
\usepackage{listings}
\usepackage{xcolor}
\usepackage{enumerate}
\usepackage{epsfig}
\usepackage{latexsym}
\usepackage{amsfonts}
\usepackage{wrapfig}
\usepackage{comment}
\usepackage{enumitem}
\usepackage{paralist}
\usepackage{cite}
\usepackage{url}
\usepackage{breakurl}

\newcommand {\ie} {{\em i.e., }}

\newcommand {\beq} {\begin{equation}}
\newcommand {\eeq} {\end{equation}}
\newcommand {\bequn} {\begin{equation*}}
\newcommand {\eequn} {\end{equation*}} 
\newcommand {\bear} {\begin{eqnarray}}
\newcommand {\eear} {\end{eqnarray}}
\newcommand {\bearun} {\begin{eqnarray*}}
\newcommand {\eearun} {\end{eqnarray*}}

\newcommand {\Eqref}[1]{Eq.~(\ref{#1})}

\newtheorem{theorem}{Theorem}
\newtheorem{lemma}[theorem]{Lemma}

\newcounter{chl}

\title{When Two is Worse Than One}
\author{R.~Gu\'{e}rin\\
Department of Computer Science \& Engineering\\
Washington University in St. Louis\\
{\tt guerin@wustl.edu}
}
\date{June 5, 2019}

\begin{document} 
\maketitle
\begin{abstract}
This note is concerned with the impact on job latency of splitting a token bucket into multiple
sub-token buckets with equal aggregate parameters and offered the same
job arrival process.  The situation commonly arises in distributed computing
environments where job arrivals are rate controlled (each job needs one token to
enter the system), but capacity limitations call for distributing jobs across
multiple compute resources with scalability considerations preventing the use
of a centralized rate control component (each compute resource is responsible
for monitoring and enforcing that the job stream it receives conforms to a certain
traffic envelope). The question we address is to what
extent splitting a token bucket into multiple sub-token buckets that individually
rate control a subset of the original arrival process affects job latency, when
jobs wait for a token whenever the token bucket is empty upon their arrival.
Our contribution is to establish that independent of the job arrival process and
how jobs are distributed across compute resources (and sub-token buckets), 
splitting a token bucket always increases the sum of job latencies in the 
token buckets, and consequently the average job latency.

\end{abstract}

\pagenumbering{arabic}

\section{Model, Assumptions, and Motivations}

Consider a two-parameters token bucket~\cite{rfc2697} $(r,b)$ where
$r$ denotes the token rate (in messages/sec) and $b$ the allowed burst
size (in jobs or messages). In other words the number of jobs that can
leave the token bucket in any time interval of duration $\Delta t$
(the arrival curve to the system downstream of the token bucket) is
upper-bounded by $b+r\cdot \Delta t$.  Job arrivals to the token
bucket follow an arbitrary arrival process and each job consumes
one token.  Jobs that find an available token upon their arrival
immediately clear the token bucket without incurring any
delay. Jobs that arrive to an empty token bucket (or a token
bucket with only a fraction of a token) wait until a full token is
available before they are allowed to leave the token bucket. The
waiting space at the token bucket is assumed large enough (infinite)
to ensure that jobs waiting for tokens are never lost.

Of concern is the latency that jobs can incur in the token bucket.
The primary job latency metric of interest is the sum of the
job latencies, or conversely the average job latency, \ie the
sum of job latencies divided by the number of jobs.
More specifically, the problem we are investigating in this note is
the impact on job latency when replacing a one-bucket system
$(r,b)$ with a two (or more) bucket system consisting of two separate
sub-token buckets $(r_1,b_1)$ and $(r_2,b_2)$, where $r=r_1+r_2$ and
$b=b_1+b_2$.  In the two-bucket system, the original stream of
arrivals is split arbitrarily across the two sub-token buckets at the
times of job arrivals, and each sub-token bucket has an infinite
queue where jobs waiting for tokens can be stored.  Since jobs
are indistinguishable, we initially assume for simplicity that they
are served in first-served-first-come (FCFS) order\footnote{Note that
  a FCFS service order is known to minimize the sum of job
  latencies in both single-server and multi-server systems when
  service times are constant~\cite{uuganbaatar11}.} in both the
one-bucket and two-bucket systems, though as we shall see the main
results hold under arbitrary service ordering.

The primary motivation for the investigation is that of Distributed
Rate Limiting (DRL) systems that arise in distributed
computing environments as found in the cloud or
datacenters~\cite{tyk,yahoo}. In such settings, users specify a job
traffic profile in the form of a token bucket, while the compute
service provider provisions resources to ensure an agreed upon Service
Level Objective (SLO) that commonly includes (average) latency.
Because of resource constraints, it is often necessary for the
provider to distribute the user's jobs across multiple compute
facilities. For scalability, rate control is performed separately at
each compute facility, which in turn calls for splitting the original
token bucket into multiple sub-token buckets, one for each compute
facility~\cite{raghavan17}.  Furthermore, ensuring that the user job
arrival process still conforms to the original traffic envelope, calls
for preserving the total job arrival rate and burst size across
sub-token buckets.

Towards investigating the performance of a DRL system, we first note
that under the assumptions of a general job arrival process with
each job requiring exactly one token, a token bucket system with
unit token rate, \ie $r=1$, and a bucket size of $b$ tokens behaves
like a modified G/D/1 queue with unit service times.  The modification
is that in the token bucket system, jobs experience a delay if and
only if the queue content in the G/D/1 system exceeds $b-1$.  In other
words, the token bucket delay $d_i$ of the $i^{th}$ job can be
obtained from the system time of this job in the corresponding
G/D/1 system as follows:
\beq
\label{eq:tb_delay}
d(a_i)=\max\{0,U(a_i^-)+1-b\}\, ,
\eeq
where $U(a_i^-)+1$ corresponds to the unfinished work found in the
G/D/1 queue by the $i^{th}$ job upon its arrival at time $a_i$
plus its own contribution to the unfinished work, and $b$ is the
bucket size.

Next, we proceed to compare the relative (latency) performance of a
one-bucket system to that of a multi-bucket system obtained by splitting
the one-bucket system as described above.  In particular, we establish
that splitting a token bucket in two (or more) sub-token buckets
always increases the sum of the job latencies, and hence the
average job latency.

\section{One vs.~Two or more Token Buckets}
\label{sec:1vs2}

Towards establishing the result that splitting a token bucket can only
worsen the sum of job latencies, we first state a simple Lemma.
\begin{lemma}
\label{lemma:uw}
At any point in time $t$, the unfinished work $U(t)$ in a
work-conserving G/D/1 queue is smaller than or equal to the total
unfinished work $U^{(k)}(t)=\sum_{l=1}^kU_l(t)$ in a set of $k$
work-conserving G/D/1 queues with the same aggregate service rate and
fed the same arrival process.
\end{lemma}
\begin{proof}
The result directly stems from the observation that when fed the same
arrival process, $k$ parallel work-conserving G/D/1 queues never clear
work faster than a single work-conserving G/D/1 queue with the same
aggregate service rate.  Specifically, at any point in time both the
one-queue and the $k$-queues system have received the same amount of
work (they are fed the same set of arrivals), both systems are
work-conserving, and the one-queue system processes work at least as
fast as the $k$-queue system whenever it is not empty, so that it can
never have more unfinished work than the $k$-queue system.

Formally, we assume that up to the start of the $j^{th}$ busy period
of the one-queue system, the unfinished work in the one-queue system
has always been smaller than or equal to that of the $k$-queue
system, and wlog we assume that the one-queue system has unit service
rate.  We establish the result by induction on the busy periods of the
one-queue system.

Denote as $t_j$ the start of the $j^{th}$ busy period of the one-queue
system, and let $T_j$ denote the duration of that busy period.  The
unfinished work in the one-queue system during that busy period is
then of the form $U(t)=U(t_j^-)+W(t_j,t)-(t-t_j)=W(t_j,t)-(t-t_j)\,
,\,\forall t\in[t_j,t_j+T_j]$, where $W(t_j,t)$ represents the amount
of work that has arrived in $[t_j,t]$, and we have used the fact that
by definition the unfinished work just before the start of a busy
period is~$0$.  Similarly, the unfinished work in the $k$-queue system
is of the form $U^{(k)}(t) = U^{(k)}(t_j^-)+W(t_j,t)-
\int_{t_j}^tr^{(k)}(u)du\geq W(t_j,t)-(t-t_j)=U(t)\, ,\,\forall
t\in[t_j,t_j+T_j]$, where we have used the facts that
$U^{(k)}(t_j^-)\geq U(t_j^-)=0$ (from our induction hypothesis),
$r^{(k)}(u)\leq 1$, \ie the aggregate service rate in the $k$-queue
system can never exceed the unit service rate of the one-queue system,
and both systems receive the same amount of work $W(t_j,t)$.
Furthermore, because by definition of a busy period $U(t_j+T_j)=0$ and
both the one-queue and the $k$-queue system see the same arrivals, we
also have $0=U(t)\leq U^{(k)}(t)\, ,\,\forall t\in[t_j+T_j,t_{j+1})$,
  where $t_{j+1}$ is the start time of the $(j+1)^{th}$ busy period of
  the one-queue system, \ie the time of the next arrival after
  $t_j+T_j$.  This establishes that the unfinished work in the
  one-queue system remains smaller than or equal to that in the
  $k$-queue system until the start of the $(j+1)^{th}$ busy period of
  the one-queue system.  This completes the proof of the induction
  step.
\end{proof}

We are now ready to state our main result, which establishes that
splitting a two-parameter token bucket $(r,b)$ into multiple sub-token
buckets $(r_l,b_l),\, l=1,\ldots,k,$ with equivalent aggregate
parameters $r=\sum_{l=1}^kr_l$ and $b=\sum_{l=1}^kb_l$, is never
beneficial when it comes to the overall (sum or average)
job latency introduced by the rate control enforcement of the
token bucket.

\begin{theorem}
\label{theo:1vs2}
Given a two-parameter token bucket $(r,b)$ and a general job
arrival process where jobs each require one token to exit the
bucket, splitting this one-bucket system into multiple, say, $k$,
sub-token buckets with parameters $(r_l,b_l)$ such that
$r=\sum_{l=1}^kr_l$ and $b=\sum_{l=1}^kb_l$, can never improve the sum
of the job latencies, irrespective of how jobs are distributed
to the $k$ sub-token buckets.  More generally, denoting as $S(t)$ and
$S^{(k)}(t)$ the sum of the delays accrued by all jobs up to time
$t$ in the one-bucket and $k$-bucket systems, respectively, we have
\beq
\label{eq:theo}
S(t)\leq S^{(k)}(t)\,, \,\forall t
\eeq
\end{theorem}
\begin{proof}
We first establish the result for the case $k=2$, and wlog assume that
$r=1$. 

The proof is simply based on the fact that jobs waiting for tokens
in either system accrue delay at the same rate, and establishing that
at any time $t$ the number $N(t)$ of jobs experiencing delays in
the one-bucket system is less than or equal to the number
$N_1(t)+N_2(t)$ of such jobs in the two-bucket system. Note that
the sum of the job delays incurred in either system up to time $t$
is of the form
\bearun
S(t) &=& \int_0^tN(u)du \\
S^{(2)}(t) &=& \int_0^t\left(N_1(u)+N_2(u)\right) du
\eearun
Hence, if $N(t)\leq N_1(t)+N_2(t)\,, \,\forall t$, then $S(t)\leq
S^{(2)}(t)\,, \,\forall t$, which proves the result for $k=2$.  We
therefore proceed to establish that $N(t)\leq N_1(t)+N_2(t)\,,
\,\forall t$.

The number $N(t)$ of jobs waiting for tokens, \ie accruing delay,
at time $t$ in a one-bucket system with bucket size $b$ is of the
the form 
\bequn
N(t)=\left\lceil \max\{0,U(t)-b\}\right\rceil \, ,
\eequn
where $\lceil x\rceil$ represents the ceiling of~$x$, $U(t)$ is the
unfinished work in the corresponding G/D/1/ queue, and consistent with
\Eqref{eq:tb_delay} we have used the fact that jobs are delayed in
the token bucket only when the unfinished work in the G/D/1 queue
exceeds the bucket size~$b$.

Similarly the total number of jobs waiting for tokens in a
two-bucket system with bucket sizes $b_1$ and $b_2$ such that
$b=b_1+b_2$ is of the form
\bequn
N_1(t)+N_2(t)=\left\lceil \max\{0,U_1(t)-b_1\}\right\rceil
             +\left\lceil \max\{0,U_2(t)-b_2\}\right\rceil
\eequn
Since we know that $\lceil x\rceil\leq \lceil x_1\rceil+\lceil
x_2\rceil$, when $x\leq x_1+x_2$, we focus on establishing that 
\beq
\label{eq:max}
\max\{0,U(t)-b\}\leq \max\{0,U_1(t)-b_1\}+\max\{0,U_2(t)-b_2\}
\eeq
From Lemma~\ref{lemma:uw}, we know that $U(t)\leq U_1(t)+U_2(t)$.
Next, we consider separately the cases $U(t)-b\leq 0$ and $U(t)-b>0$.

\noindent
{\bf Case~$1$}: $U(t)-b\leq 0$

In this case, \Eqref{eq:max} is trivially verified.

\noindent
{\bf Case~$2$}: $U(t)-b>0$

We further separate this case in two separate sub-cases:

\noindent
{\bf Case~$2a$}: $U_1(t)-b_1\leq 0$ and $U_2(t)-b_2\geq 0$ (or interchangeably
$U_1(t)-b_1\geq 0$ and $U_2(t)-b_2\leq 0$) 

In this case, \Eqref{eq:max} simplifies to
\bequn
U(t)-b\leq U_2(t)-b_2
\eequn
Applying again the result of Lemma~\ref{lemma:uw}, we have
\bequn
U(t)\leq U_1(t)+U_2(t)\,\Rightarrow \, U(t)-b\leq
U_1(t)+U_2(t)-b_1-b_2\, \Rightarrow \, U(t)-b\leq U_2(t)-b_2\, ,
\eequn
where we have used the fact that $b=b_1+b_2$ and $U_1(t)-b_1\leq
0$. Hence, \Eqref{eq:max} again holds in Case~$2a$.

\noindent
{\bf Case~$2b$}: $U_1(t)-b_1\geq 0$ and $U_2(t)-b_2\geq 0$

In this case, \Eqref{eq:max} becomes
\bequn
U(t)-b\leq U_1(t)-b_1+U_2(t)-b_2\, ,
\eequn
which again holds because of Lemma~\ref{lemma:uw} and the fact that
$b=b_1+b_2$. 

Since the case $U_1(t)-b_1\leq 0$ and $U_2(t)-b_2\leq 0$ is not
possible under Case~$2$ (it would violate Lemma~\ref{lemma:uw}), this
establishes that \Eqref{eq:max} holds in all cases.  Accordingly,
$N(t)\leq N_1(t)+N_2(t)\,, \,\forall t$, so that as mentioned earlier,
$S(t)\leq S^{(2)}(t)\,, \,\forall t$, which establishes the result for
$k=2$.

Extending the result to $k>2$ is readily accomplished by applying the
above approach recursively to groups of two sub-token buckets.
\end{proof}

In concluding, we note that while \Eqref{eq:tb_delay} assumed an FCFS
service ordering for jobs in the token bucket, both
Lemma~\ref{lemma:uw} and Theorem~\ref{theo:1vs2} are independent of
the order in which jobs waiting for tokens are scheduled for
transmission, as long as the schedule is ``work-conserving,'' \ie
jobs (any waiting job) leave as soon as one full token is
available. In other words, available tokens are not split across
multiple waiting jobs.

\section*{Acknowledgment}
This work was supported by NSF grant CNS 1514254. Any opinions,
findings, and conclusions or recommendations expressed in this
material are those of the author and do not necessarily reflect the
views of the National Science Foundation.

\end{document}